\documentclass{article}

\usepackage{fullpage}
\usepackage{amsthm,amssymb,amsmath}  
\usepackage{xspace,enumerate}
\usepackage[utf8]{inputenc}
\usepackage{thmtools}
\usepackage{thm-restate}
\usepackage{authblk}

  \theoremstyle{plain}
  \newtheorem{theorem}{Theorem}
  \newtheorem{lemma}{Lemma}  
  \newtheorem{corollary}[theorem]{Corollary}  
  \newtheorem{fact}{Fact}
  \newtheorem{observation}{Observation}
  \theoremstyle{definition}
  
  \newtheorem{problem}{Problem}

\title{Longest common substring with approximately $k$ mismatches\thanks{This is a full and extended version of the conference paper \cite{DBLP:conf/cpm/Starikovskaya16}.}}
\author[1]{Tomasz Kociumaka}
\author[1]{Jakub Radoszewski}
\author[2]{Tatiana Starikovskaya}

\affil[1]{Institute of Informatics, University of Warsaw, Warsaw, Poland\\
    \texttt{[kociumaka,jrad]@mimuw.edu.pl}}
\affil[2]{DI/ENS, PSL Research University, Paris, France\\
    \texttt{tat.starikovskaya@gmail.com}
}

\date{\vspace{-5ex}}

\usepackage{algorithm,algorithmicx,algpseudocode}
\usepackage{amsthm,amsfonts,amsmath}
\usepackage{microtype}
\usepackage{xspace}

\newtheorem*{hypothesis}{\bf Hypothesis}

\newcommand{\Oh}{O}
\newcommand{\eps}{\varepsilon}
\newcommand{\lcpe}{\mathrm{LCP}_{(1+\eps)k}}
\newcommand{\lcp}{\mathrm{LCP}_{\tilde{k}}}
\newcommand{\lcpk}{\mathrm{LCP}_{k}}
\newcommand{\sk}{\mathrm{sk}}
\newcommand{\Prob}{\mathrm{Pr}}
\newcommand{\LCSp}{\textsf{LCS}\xspace}
\newcommand{\kLCS}{\textsf{LCS with $k$ Mismatches}\xspace}
\newcommand{\AllkLCS}{\textsf{LCS with $k$ Mismatches for all $k$}\xspace}
\newcommand{\kApproxLCS}{\textsf{LCS with Approximately $k$ Mismatches}\xspace}
\newcommand{\ApproxLCS}{\textsf{Approximate LCS with $k$ Mismatches}\xspace}
\newcommand{\OV}{\textsf{Orthogonal Vectors}\xspace}
\newcommand{\BJI}{\textsf{Binary Jumbled Indexing}\xspace}

\newcommand{\Hashes}{\mathcal{H}}
\newcommand{\Collisions}{\mathcal{C}}
\newcommand{\Bad}{\mathcal{B}}
\newcommand{\Projections}{\Pi}
\newcommand{\Pos}{{P}}
\newcommand{\HD}{d_H}

\newcommand\restr[2]{{
  \left.\kern-\nulldelimiterspace
  #1 
  \vphantom{\big|} 
  \right|_{#2} 
}}

\begin{document}
  \maketitle

\begin{abstract}
In the longest common substring problem, we are given two strings of length $n$ and must find a substring of maximal length that occurs in both strings. It is well known that the problem can be solved in linear time, but the solution is not robust and can vary greatly when the input strings are changed even by one character. To circumvent this, Leimeister and Morgenstern introduced the problem of the longest common substring with $k$ mismatches. Lately, this problem has received a lot of attention in the literature. In this paper, we first show a conditional lower bound based on the SETH hypothesis implying that there is little hope to improve existing solutions. We then introduce a new but closely related problem of the longest common substring with approximately $k$ mismatches and use locality-sensitive hashing to show that it admits a solution with strongly subquadratic running time. We also apply these results to obtain a strongly subquadratic-time 2-approximation algorithm for the longest common substring with $k$ mismatches problem and show conditional hardness of improving its approximation ratio.
\end{abstract}

\section{Introduction}
\label{sec:intro}
Understanding how similar two strings are and what they share in common is a central task in stringology. The significance of this task is witnessed by the 50,000+ citations of the paper introducing BLAST~\cite{AGM+90}, a heuristic algorithmic tool for comparing biological sequences. This task can be formalised in many different ways, from the longest common substring problem to the edit distance problem. The longest common substring problem can be solved in optimal linear time and space, while the best known algorithms for the edit distance problem require $n^{2-o(1)}$ time, which makes the longest common substring problem an attractive choice for many practical applications. On the other hand, the longest common substring problem is not robust and its solution can vary greatly when the input strings are changed even by one character. To overcome this issue, recently a new problem has been introduced called the longest common substring with $k$ mismatches. In this paper, we continue this line of research. 

\subsection{Related work}
Let us start with a precise statement of the longest common substring problem.

\begin{problem}[\LCSp ]~\label{pr:LCS}
Given two strings $T_1, T_2$ of length $n$, find a maximum-length substring of $T_1$ that occurs in $T_2$.
\end{problem}

The suffix tree of $T_1$ and $T_2$, a data structure containing all suffixes of $T_1$ and $T_2$, allows to solve this problem in linear time and space~\cite{DBLP:conf/focs/Weiner73,DBLP:books/cu/Gusfield1997,DBLP:conf/cpm/Hui92}, which is optimal as any algorithm needs $\Omega(n)$ time to read and $\Omega(n)$ space to store the strings. However, if we only account for ``additional'' space, the space the algorithm uses apart from the space required to store the input, then the suffix tree-based solution is not optimal and has been improved in a series of publications~\cite{LCS-1,DBLP:conf/esa/KociumakaSV14,DBLP:conf/cpm/StarikovskayaV13}. 

The major disadvantage of the longest common substring problem is that its solution is not robust. Consider, for example, two pairs of strings: $a^{2m+1}, a^{2m} b$ and $a^{m} b a^{m},a^{2m} b$. The longest common substring of the first pair of strings is almost twice as long as the longest common substring of the second pair of strings, although we changed only one character. This makes the longest common substring unsuitable to be used as a measure of similarity of two strings: Intuitively, changing one character must not change the measure of similarity much. To overcome this issue, it is natural to allow the substring to occur in $T_1$ and $T_2$ not exactly but with a small number of mismatches. 

\begin{problem}[\kLCS]\label{pr:LCSk}
Given two strings $T_1, T_2$ of length $n$ and an integer $k$, find a maximum-length substring of $T_1$ that occurs in $T_2$ with at most $k$ mismatches.
\end{problem}

The problem can be solved in quadratic time and space by a dynamic-program\-ming algorithm, but more efficient solutions have also been shown. The longest common substring with one mismatch problem was first considered in~\cite{DBLP:journals/poit/BabenkoS11}, where an $\Oh(n^2)$-time and $\Oh(n)$-space solution was given. This result was further improved by Flouri et al.~\cite{DBLP:journals/ipl/FlouriGKU15}, who showed an $\Oh(n \log n)$-time and $\Oh(n)$-space solution. 

For a general value of $k$, the problem was first considered by Leimeister and Morgenstern~\cite{kmacs}, who suggested a greedy heuristic algorithm. Flouri et al.~\cite{DBLP:journals/ipl/FlouriGKU15} showed that \kLCS admits a quadratic-time algorithm which takes constant (additional) space. Grabowski~\cite{DBLP:journals/ipl/Grabowski15} presented two output-dependent algorithms with running times $\Oh (n ((k+1) (\ell_0+1))^k)$ and $\Oh (n^2 k/\ell_k)$, where $\ell_0$ is the length of the longest common substring of $T_1$ and $T_2$ and $\ell_k$ is the length of the longest common substring with $k$ mismatches of $T_1$ and $T_2$. Thankachan et al.~\cite{DBLP:journals/jcb/ThankachanAA16} gave an $\Oh(n)$-space, $\Oh(n \log^k n)$-time solution for $k=\Oh(1)$. Very recently, Charalampopoulos et al.~\cite{DBLP:conf/cpm/Charalampopoulos18} extended the underlying techniques and developed an $\Oh(n)$-time 
algorithm for the case of $\ell_k = \Omega(\log^{2k+2} n)$.
Finally, Abboud et al.~\cite{DBLP:conf/soda/AbboudWY15} applied the polynomial method to develop a $k^{1.5} n^2 / 2^{\Omega(\sqrt{(\log n)/k})}$-time randomised solution to the problem. In fact, their algorithm was developed for a more general problem of computing the longest common substring with $k$ edits, but it can be adapted to \kLCS as well. The problem of computing the longest common substring with $k$ edits was also considered in~\cite{DBLP:conf/recomb/ThankachanACA18}, where an $\Oh(n \log^k n)$-time solution was given for constant $k$.

\subsection{Our contribution}
Our contribution is as follows. In Section~\ref{sec:hardness}, we show that existence of a strongly subquadratic-time algorithm for \kLCS on strings over binary alphabet for $k = \Omega(\log n)$ refutes the Strong Exponential Time Hypothesis (SETH) of Impagliazzo, Paturi, and Zane~\cite{DBLP:journals/jcss/ImpagliazzoP01,DBLP:journals/jcss/ImpagliazzoPZ01};
see also~\cite[Chapter 14]{ksiazka}:

\begin{hypothesis}[SETH]
For every $\delta > 0$, there exists an integer $q$ such that SAT on $q$-CNF formulas with $m$ clauses and  $n$ variables cannot be solved in $m^{O(1)} 2^{(1-\delta) n}$ time.
\end{hypothesis}

This conditional lower bound implies that there is little hope to improve existing solutions to \kLCS. To this end, we  introduce a new problem, inspired by the work of Andoni and Indyk~\cite{substringNN}. 

\begin{problem}[\kApproxLCS]\label{pr:LCS'k}
Two strings $T_1, T_2$ of length $n$, an integer $k$, and a constant $\eps > 0$ are given. If $\ell_k$ is the length of the longest common substring with $k$ mismatches of $T_1$ and $T_2$, return a substring of $T_1$ of length at least $\ell_k$ that occurs in $T_2$ with at most $(1+\eps) \cdot k$ mismatches.
\end{problem}

Let $d_H(S_1,S_2)$ denote the Hamming distance between equal-length strings $S_1$ and $S_2$, that is, the number of mismatches between them. Then we are to find the substrings $S_1$ and $S_2$ of $T_1$ and $T_2$, respectively, of length at least $\ell_k$ such that $d_H(S_1,S_2) \le (1+\eps) \cdot k$.

Although the problem statement is not standard, it makes perfect sense from the practical point of view. It is also more robust than the \kLCS problem, as for most applications it is not important whether a returned substring occurs in $T_1$ and $T_2$ with, for example, $10$ or $12$ mismatches. The result is also important from the theoretical point of view as it improves our understanding of the big picture of string comparison. In their work, Andoni and Indyk used the technique of locality-sensitive hashing to develop a space-efficient randomised index for a variant of the approximate pattern matching problem. We extend their work with new ideas in the construction and the analysis to develop a randomised subquadratic-time solution to Problem~\ref{pr:LCS'k}. This result is presented in Section~\ref{sec:kApproxLCS}.

In Section~\ref{sec:ApproxLCS}, we consider approximation algorithms for the length of the \kLCS. By applying previous techniques, we show a strongly subquadratic-time 2-approximation algorithm and show that no strongly subquadratic-time $(2-\eps)$-approximation algorithm exists for any $\eps>0$ unless SETH fails.

Finally, in Section~\ref{sec:AllkLCS} we show a strongly subcubic-time solution for \AllkLCS by reducing it (for arbitrary alphabet size) to \BJI. Namely, we show that \kLCS for all $k = 1, \dots, n$ can be solved in $\Oh(n^{2.859})$ expected time or in $\Oh(n^{2.864})$ deterministic time, improving upon naive computation performed for every $k$ separately.

\section{\kLCS is SETH-hard}
\label{sec:hardness}
Recall that the Hamming distance of two strings $U$ and $V$ of the same length, denoted as $d_H(U,V)$, is simply the number of mismatches. Our proof is based on conditional hardness of the following problem.

\begin{problem}[\OV]~\label{pr:OV}
Given a set $A$ of $N$ vectors from $\{0,1\}^d$ each, does there exist a pair of vectors $U,V \in A$ that is orthogonal, i.e., $\sum_{h=1}^d U[h] V[h]=0$?
\end{problem}

Williams showed a conditional lower bound for an equivalent problem called cooperative subset queries~\cite[Section 5.1]{DBLP:journals/tcs/Williams05}, which immediately implies the following fact:

\begin{fact}\label{fct:OVP}
Suppose there is $\varepsilon>0$ such that for all constant $c$, \OV on a set of $N$ vectors of dimension $d=c\log N$ can be solved in $2^{o(d)} \cdot N^{2-\varepsilon}$ time. Then SETH is false.
\end{fact}

We treat vectors from $\{0,1\}^d$ as binary strings of length $d$. Let us introduce two morphisms, $\mu$ and $\tau$:

$$\mu(0)={011}\,1000, \quad \mu(1)={000}\,1000, \quad \tau(0)={001}\,1000, \quad \tau(1)={111}\,1000.$$
We will use the following two observations.

\begin{observation}\label{obs:11}
We have $d_H(\mu(0),\tau(0))=d_H(\mu(0)$, $\tau(1))=d_H(\mu(1),\tau(0))=1$, and $d_H(\mu(1),\tau(1))=3$.
\end{observation}

\begin{observation}\label{obs:1000}
Let $x,y,z \in \{0,1\}$. Then the string $1000$ has exactly two occurrences in $1000xyz1000$.
\end{observation}

Let us also introduce a string gadget $H=\gamma^d$, where $\gamma=100\,1000$. Note that $\gamma \ne \mu(x)$ and $\gamma \ne \tau(x)$ for $x \in \{0,1\}$. Further, note that $|H|=|\mu(U)|=|\tau(U)|=7d$ for any $U \in A$.

\begin{lemma}\label{lem:red}
  Consider a set of vectors $A=\{U_1,\dots,U_N\}$ from $\{0,1\}^d$, the strings:
  $$T_1=H^q \mu(U_1) H^q \cdots \mu(U_N) H^q, \quad T_2=H^q \tau(U_1) H^q \cdots \tau(U_N) H^q$$
  for some positive integer $q$, and $k=d$. Then:
  \begin{enumerate}[(a)]
    \item\label{it:a} If the set $A$ contains two orthogonal vectors, then the \kLCS problem for $T_1$ and $T_2$
    has a solution of length at least $\ell=(14q+7)d$.
    \item\label{it:b} If the set $A$ does not contain two orthogonal vectors, then all the solutions for the \kLCS problem for $T_1$ and $T_2$
    have length smaller than $\ell'=(7q+14)d$.
  \end{enumerate}
\end{lemma}
\begin{proof}
\eqref{it:a} Assume that $U_i$ and $U_j$ are a pair of orthogonal vectors.
$T_1$ contains a substring $H^q \mu(U_i) H^q$ and $T_2$ contains a substring $H^q \tau(U_j) H^q$. Both substrings have length $\ell$ and, by Observation~\ref{obs:11}, their Hamming distance is exactly $k = d$.

\eqref{it:b} Assume to the contrary that there are indices $a$ and $b$ for which the substrings $S_1=T_1[a,a+\ell'-1]$ and $S_2=T_2[b,b+\ell'-1]$ have at most $k$ mismatches. First, let us note that $7 \mid a-b$. Indeed, otherwise $S_1$ would contain at least $\lfloor (\ell'-3)/7 \rfloor=(q+2)k-1\ge k+1$ substrings of the form $1000$ which, by Observation~\ref{obs:1000}, would not be aligned with substrings $1000$ in $S_2$. Hence, they would account for more than $k$ mismatches between $S_1$ and $S_2$.

Let us call all the substrings of $T_1$ and $T_2$ that come from the 3-character prefixes of $\mu(0)$, $\mu(1)$, $\tau(0)$, $\tau(1)$, and $\gamma$ the \emph{core substrings}, with core substrings that come from $\gamma$ being \emph{gadget core substrings}. We have already established that the core substrings of $S_1$ and $S_2$ are aligned. Moreover, $S_1$ and $S_2$ contain at least $\lfloor (\ell'-2)/7 \rfloor=(q+2)k-1$ core substrings each. Amongst every $(q+2)k-1$ consecutive core substrings in $S_1$, some $k$ consecutive must come from $\mu(U_i)$ for some index $i$; a symmetric property holds for $S_2$ and $\tau(U_j)$. Moreover, as only the gadget core substrings in $S_1$ and $S_2$ can match exactly, at most $k$ core substrings that are contained in $S_1$ and $S_2$ can be non-gadget. Hence, $S_1$ and $S_2$ contain exactly $k$ non-gadget core substrings each. If they were not aligned, they would have produced more than $k$ mismatches in total with the gadget core substrings.

Therefore, $S_1$ and $S_2$ must contain, as aligned substrings, $\mu(U_i)[1,7d-4]$ and $\tau(U_j)[1,7d-4]$ for some $i,j \in \{1,\ldots,N\}$, respectively. Hence, $d_H(U_i, U_j) \le k$. By Observation~\ref{obs:11}, we conclude that $U_i$ and $U_j$ are orthogonal.
\end{proof}
  
\begin{theorem}
Suppose there is $\eps > 0$ such that \kLCS can be solved in $\Oh(n^{2-\varepsilon})$ time on strings over binary alphabet for $k = \Omega(\log n)$. Then SETH is false.
\end{theorem}
\begin{proof}
The reduction of Lemma~\ref{lem:red} with $q=1$ constructs, for an instance of the \OV problem with $N$ vectors of dimension $d$, an equivalent instance of the \kLCS problem with strings of length $n=7d(2N+1)$ and $k=d$. Thus, assuming that \kLCS can be solved in $\Oh(n^{2-\eps})$ time for $k = \Omega(\log n)$, the constructed instance can be solved in $\Oh(N^{2-\eps}d^{\Oh(1)})$ time if $d = c \log N$. This, by Fact~\ref{fct:OVP}, contradicts SETH.
\end{proof}

\section{\kApproxLCS}
\label{sec:kApproxLCS}	
	In this section, we prove the following theorem. 

\begin{theorem}\label{th:main}
Let $\eps \in (0,2)$ and $\delta \in (0,1)$ be arbitrary constants. The \kApproxLCS problem can be solved in $\Oh(n^{1+1/(1+\eps)} )$ space and $\Oh (n^{1+1/(1+\eps)} \log^2 n)$ time with error probability $\delta$.
\end{theorem}

\subsection{Overview of the proof}
	The classic solution to the longest common substring problem is based on two observations. The first observation is that the longest common substring of $T_1$ and $T_2$ is in fact the longest common prefix of some suffix of $T_1$ and some suffix of $T_2$. The second observation is that the maximal length of the longest common prefix of a fixed suffix $S$ of $T_1$ and suffixes of $T_2$ is reached by one of the two suffixes of $T_2$ that are closest to $S$ in the lexicographic order. This suggests the following algorithm: First, build a suffix tree of $T_1$ and $T_2$, which contains all suffixes of $T_1$ and $T_2$ ordered lexicographically. Second, compute the longest common prefix of each suffix of $T_1$ and the two suffixes of $T_2$ closest to $S$ in the lexicographic order, one from the left and one from the right. The problem of computing the longest common prefix has been extensively studied in the literature and a number of very efficient deterministic and randomised solutions exist~\cite{LCE-1,DBLP:journals/jda/BilleGSV14,DBLP:journals/siamcomp/FischerH11,LCE-4,DBLP:journals/siamcomp/HarelT84}; for example, one can use a Lowest Common Ancestor (LCA) data structure, which can be constructed in linear time and space and answers longest common prefix queries in $\Oh(1)$ time~\cite{DBLP:journals/siamcomp/FischerH11,DBLP:journals/siamcomp/HarelT84}.

Our solution to the longest common substring with approximately $k$ mismatches problem is somewhat similar. Instead of the lexicographic order, we will consider $\Theta (n^{1/(1+\eps)})$ different orderings on the suffixes of $T_1$ and~$T_2$. To define these orderings, we will use the locality-sensitive hashing technique, which was initially introduced for the needs of computational geometry~\cite{DBLP:journals/toc/Har-PeledIM12} and later adapted for substrings with Hamming distance~\cite{substringNN}. In more detail, we will choose $\Theta (n^{1/(1+\eps)})$ hash functions, where each function can be considered as a projection of a string of length $n$ onto a random subset of its positions. By choosing the size of the subset appropriately, we will be able to guarantee that the hash function is locality-sensitive: For any two strings at the Hamming distance at most $k$, the values of the hash functions on them will be equal with reasonably high probability, while the values of the hash functions on any pair of strings at the Hamming distance bigger than $(1+\eps) \cdot k$ will be equal with low probability.
For each hash function, we will sort the suffixes of $T_1$ and $T_2$ by the lexicographic order on their hash values. As a corollary of the locality-sensitive property, if two suffixes of $T_1$ and $T_2$ have a long common prefix with at most $k$ mismatches, they are likely to be close to each other in at least one of the orderings.

However, we will not be able to compute the longest common prefix with $(1+\eps) k$ mismatches for all candidate pairs of suffixes exactly (the best data structure, based on the kangaroo method~\cite{DBLP:journals/tcs/LandauV86,kangaroo-2}, has query time $\Theta((1+\eps)k)$ which is $\Theta(n)$ in the worst case). We will use this method for only one pair of suffixes chosen at random from a carefully preselected set of candidate pairs. For other candidate pairs, we will use $\lcp$ queries.  In an $\lcp$ query, we are given two suffixes $S_1, S_2$ of $T_1$ and $T_2$, respectively, and must output any integer $\ell$ such that $\lcpk(S_1,S_2)\le \ell \le  \lcpe(S_1,S_2)$, where 
$\lcpk$ and $\lcpe$ denote the longest common prefix with at most $k$ and at most $(1+\eps)k$ mismatches, respectively.
In Section~\ref{sec:kLCP}, we show the following lemma based on the sketching techniques by Kushilevitz et al.~\cite{DBLP:journals/siamcomp/KushilevitzOR00}:

\begin{lemma}\label{lm:lcp_k}
For given $k$ and $\eps$, after $\Oh (n \log^3 n)$-time and $\Oh (n \log^2 n)$-space preprocessing of strings $T_1,T_2$, any $\lcp$  query can be answered in $\Oh(\log^2 n)$ time. 
With probability at least $1-1/n^3$, the preprocessing produces a data structure that correctly answers all $\lcp$ queries.
\end{lemma}

The key idea is to compute sketches for all power-of-two length substrings of $T_1$ and $T_2$. The sketches will have logarithmic length (so that we will be able to compare them very fast) and the Hamming distance between them will be roughly proportional to the Hamming distance between the original substrings. Once the sketches are computed, we use binary search to answer $\lcp$ queries in polylogarithmic time.

\subsection{Proof of Lemma~\ref{lm:lcp_k}}
\label{sec:kLCP}
	During the preprocessing stage, we compute sketches~\cite{DBLP:journals/siamcomp/KushilevitzOR00} of all substrings of the strings $T_1$ and $T_2$ of lengths $\ell = 1, 2, 4, \ldots, 2^{\lfloor \log n \rfloor}$, which can be defined in the following way. Without loss of generality, assume that the alphabet is $\Sigma = \{0,1,\ldots,p-1\}$, where $p$ is a prime number. For a fixed $\ell$, choose $\lambda = \lceil 3\ln n / \gamma^2 \rceil$ vectors $r_\ell^i$ of length $\ell$, where $\gamma$ is a constant to be defined later, such that
the values $r_\ell^i [j]$ across $i = 1, 2, \dots, \lambda$ and $j = 1, 2, \dots, \ell$ are independent and identically distributed so that for every $a\in \Sigma$:
$$\Prob[r_\ell^i [j] = a] =
\begin{cases}
1-\frac{p-1}{2kp} & \mbox { if } a = 0,\\
\frac{1}{2kp} & \mbox{ otherwise. }
\end{cases}
$$
For a string $X$ of length $\ell$, we define the sketch $\sk (X)$ to be a vector of length $\lambda$, where $\sk (X) [i] = r_\ell^i \cdot X \pmod p$. For each $i = 1, 2, \ldots, \lambda$, we compute the inner product of $r_\ell^i$ with all length-$\ell$ substrings of $T_1$ and $T_2$ in $\Oh(n \log n)$ time by running the Fast Fourier Transform (FFT) algorithm in the field $\mathbb{Z}_p$~\cite{FischerPaterson}. As a result, we obtain the sketches of each length-$\ell$ substring of $T_1$ and $T_2$. We repeat this step for all specified values of $\ell$. One instance of the FFT algorithm takes $\Oh(n \log n)$ time, and we run an instance for each $i = 1, 2, \ldots, \lambda$ and for each $\ell = 1, 2, 4, \ldots, 2^{\lfloor \log n \rfloor}$, which takes $\Oh (n \log^3 n)$ time in total. The sketches occupy $\Oh (n \log^2 n)$ space. Each string $S$ can be decomposed uniquely as $X_1 X_2 \ldots X_g$, where $g = \Oh(\log n)$ and $|X_1| > |X_2| > \ldots > |X_g|$ are powers of two; we define a sketch $\sk (S) = \sum_q \sk (X_q) \pmod p$.
Let $\delta_1 = \frac{p-1}{p} (1 - (1-\frac{1}{2k})^{k})$ and $\delta_2 = \frac{p-1}{p} (1 - (1-\frac{1}{2k})^{(1+\eps) \cdot k})$.

\begin{lemma}[see \cite{DBLP:journals/siamcomp/KushilevitzOR00}]\label{lm:success_prob}
Let $S_1,S_2$ be strings of the same length. For each $i=1,\ldots,\lambda$:
\begin{itemize}
  \item if $d_H(S_1, S_2) \le k$, then $\sk (S_1)[i] \neq \sk (S_2)[i]$ with probability at most $\delta_1$;
  \item if $d_H(S_1, S_2) \ge (1+\eps) \cdot k$, then $\sk (S_1)[i] \neq \sk (S_2)[i]$ with probability at least $\delta_2$.
\end{itemize}
\end{lemma}
\begin{proof}
We use a different interpretation of $r_\ell^i$ that defines the same distribution. We start with the zero vector and sample positions with probability $\frac{1}{2k}$.
For each sampled position $j$, we decide on the value $r_{\ell}^i[j]\in \Sigma$ independently and uniformly at random. Let $m = d_H(S_1, S_2)$ and $a_1, \ldots, a_m$ be the positions of the mismatches between the two strings. If none of the positions $a_1, \ldots, a_m$ are sampled, then $\sk (S_1)[i] = \sk (S_2)[i]$. Otherwise,
if $a_{j_1},\ldots,a_{j_g}$ are sampled, for each $r_{\ell}^i[a_{j_1}], \ldots, r_{\ell}^i[a_{j_{g-1}}]$ exactly one of the $p$ choices of $r_{\ell}^i[a_{j_g}]$ results in $\sk (S_1)[i] = \sk (S_2)[i]$ (because $p$ is prime). Hence, the probability that $\sk (S_1)[i] \neq \sk (S_2)[i]$ is equal to $\frac{p-1}{p}(1-(1 - \frac{1}{2k})^{m})$, which is at most $\delta_1$ if $d_H(S_1, S_2) \le k$, and at least $\delta_2$ if the Hamming distance is at least $(1+\eps) \cdot k$.
\end{proof}

We set $\Delta = \frac{\delta_1+\delta_2}{2} \cdot \lambda$
and $\gamma = \frac{\delta_2-\delta_1}{2}$.
Observe that 

$$\gamma  = \tfrac{p-1}{2p}(1-\tfrac{1}{2k})^k\left(1-(1-\tfrac{1}{2k})^{\eps k}\right)\ge \tfrac18 \left(1-e^{-\eps/2}\right)=\Omega(\eps^{-1})$$

\noindent
because $(1-\frac{1}{2k})^k$ is an increasing function of $k$ bounded from above by $e^{1/2}$. Consequently, if $\eps$ is a constant, then $\gamma$ is a constant as well. 

\begin{lemma}\label{lm:klcperr}
For all strings $S_1$ and $S_2$ of the same length,
the following claims hold with probability at least $1-n^{-6}$:
\begin{itemize}
  \item if $d_H(\sk (S_1),\sk (S_2))>\Delta$, then $d_H(S_1, S_2)>k$;
  \item if $d_H(\sk (S_1),\sk (S_2)) \le \Delta$, then $d_H(S_1,S_2) < (1+\eps) \cdot k$.
\end{itemize}
\end{lemma}
\begin{proof}
Let $\chi_i$ be an indicator random variable that is equal to one if and only if $\sk (S_1) [i] \ne \sk (S_2) [i]$. The claim follows immediately from Lemma~\ref{lm:success_prob} and the following Chernoff--Hoeffding bounds~\cite[Theorem 1]{MR0144363}. For $\lambda$ independently and identically distributed binary variables $\chi_1, \chi_2, \ldots, \chi_\lambda$, we have

  $$\Prob \left[\frac{1}{\lambda}\sum_{i=1}^\lambda \chi_i > \mu+\gamma\right] \le e^{-2 \lambda \gamma^2}\quad\quad\text{and}\quad\quad \Prob \left[\frac{1}{\lambda}\sum_{i=1}^\lambda \chi_i \le \mu-\gamma\right] \le e^{-2\lambda \gamma^2},$$
where $\mu = \Prob [\chi_i = 1]$. Recall that $\gamma = \frac{\delta_2 - \delta_1}{2}$, so we obtain that the error probability is at most 
$e^{-2 \lambda \gamma^2} \le n^{-6}$. 

If $d_H(S_1,S_2) \le k$, Lemma~\ref{lm:success_prob} asserts that $\mu \le \delta_1$. By the first of the above inequalities, we have that $d_H(\sk (S_1),\sk (S_2)) \le \Delta$ with probability at least $1-n^{-6}$. Hence, if $d_H(\sk (S_1),\sk (S_2)) > \Delta$, then $d_H(S_1,S_2) > k$ with the same probability.

If $d_H(S_1,S_2) \ge (1+\eps) \cdot k$, Lemma~\ref{lm:success_prob} asserts that $\mu \ge \delta_2$. By the second inequality, we have that $d_H(\sk (S_1),\sk (S_2)) > \Delta$ with probability at least $1-n^{-6}$. Hence, if $d_H(\sk (S_1),\sk (S_2)) \le \Delta$, then $d_H(S_1,S_2) < (1+\eps)\cdot k$ with the same probability.
\end{proof}

Suppose we wish to answer an $\lcp$ query on two suffixes $S_1, S_2$. It suffices to find the longest prefixes of $S_1, S_2$ such that the Hamming distance between their sketches is at most~$\Delta$. As mentioned above, these prefixes can be represented uniquely as a concatenation of strings of power-of-two lengths $\ell_1 > \ell_2 > \ldots > \ell_g$. To compute $\ell_1$, we initialise it with the biggest power of two not exceeding $n$ and compute the Hamming distance between the sketches of the corresponding substrings. If it does not exceed $\Delta$, we have found $\ell_1$; otherwise, we divide $\ell_1$ by two and continue. Suppose that we already know $\ell_1, \ell_2, \ldots, \ell_i$ and the sketches $\sk(S_1[1,d_i])$ and $\sk(S_2[1,d_i])$,
where $d_i = \ell_1+\cdots + \ell_i$. To determine $\ell_{i+1}$, we initialise it with $\frac12\ell_i$ and then divide it by two until $\HD(\sk(S_1[1,d_i+\ell_{i+1}]),\sk(S_2[1,d_i+\ell_{i+1}]))\le \Delta$. These two sketches can be computed in $\Oh(\lambda)=\Oh(\log n)$ time by combining $\sk(S_1[1,d_i])$ and $\sk(S_2[1,d_i])$ with the precomputed sketches
$\sk(S_1[d_i+1,d_i+\ell_{i+1}])$ and $\sk(S_2[d_i+1,d_i+\ell_{i+1}])$, respectively.
Consequently, the query procedure takes $\Oh(\log^2 n)$ time.
It errs on at least one query with probability at most $n^{-3}$ (Lemma~\ref{lm:klcperr} is only applied for pairs of same-length substrings of $T_1$ and $T_2$, so we estimate error probability by the union bound). This completes the proof of Lemma~\ref{lm:lcp_k}.

\subsection{Proof of Theorem~\ref{th:main}}	
We start by preprocessing $T_1$ and $T_2$ as described in Lemma~\ref{lm:lcp_k}. 
In the main phase of the algorithm, we construct a family $\Hashes$ of hash functions
based on four parameters $m,s,t,w\in \mathbb{Z}$ to be specified later.

Let~$\Projections$ be the set of all projections of strings of length $n$ onto a single position, i.e.\ the value $\pi_i(S)$ of the $i$-th projection on a string $S$ is simply its $i$-th character $S[i]$.
More generally, for a string $S$ of length $n$ and a function $h=(\pi_{a_1},\ldots,\pi_{a_q}) \in \Projections^q$, we define $h(S)$ as $S[a_{p_1}] S[a_{p_2}] \cdots S[a_{p_q}]$,
where $p$ is a permutation such that $a_{p_1}\le \cdots \le a_{p_q}$.
If $|S|<n$, we define $h(S) := h(S \cdot \$^{n-|S|})$, where $ \$\notin \Sigma$ is a special gap-filling character.

Each hash function $h\in \Hashes$ is going to be a uniformly random element of $\Projections^{mt}$; however, the individual hash functions are \emph{not} chosen independently
in order to ensure faster running time for the algorithm.
Nevertheless, $\Hashes$ will be composed of $s$ independent subfamilies $\Hashes_i$, each of size $\binom{w}{t}$.
To construct $\Hashes_i$, we choose $w$ functions $u_{i,1},\ldots,u_{i,w}\in \Projections^m$ independently and uniformly at random.
Each hash function $h\in \Hashes_i$ is defined as an unordered $t$-tuple of distinct functions $u_{i,r}$.
Formally, $$\Hashes_i = \{  (u_{i,r_1}, u_{i,r_2}, \ldots, u_{i,r_t}) \in \Projections^{mt} : 1 \le r_1 < r_2 < \cdots < r_t \le w\}.$$

Consider the set of all suffixes $S_1, S_2, \ldots, S_{2n}$ of $T_1$ and $T_2$.
For each $h\in \Hashes$, we define an ordering $\prec_h$ of the suffixes $S_1,\ldots,S_{2n}$ according to the lexicographic order of the values $h(S_j)$ of the hash function and, in case of ties, according to the lengths $|S_j|$.
To construct it, we build a compact trie\footnote{Recall that a compact trie stores only explicit nodes, that is, the root, the leaves, and nodes with at least two children. Its size is linear in the number of strings that are stored. Henceforth we call a compact trie simply a trie.} on strings $h(S_1), h(S_2), \ldots, h(S_{2n})$.

\begin{theorem}\label{th:sort}
Functions $u_{i,r}$ for $i=1,\ldots,s$ and $r=1,\ldots,w$ can be preprocessed in  $\Oh(n^{4/3}\log ^{4/3} n)$ time and $\Oh(n)$ space
each, i.e., in $\Oh(sw n^{4/3}\log ^{4/3} n)$ time and $\Oh(sw n)$ space in total, so that afterwards, for each $h\in \Hashes$, a trie on $h(S_1), \ldots, h(S_{2n})$ can be constructed in $\Oh(tn\log n)$ time and $\Oh(n)$ space.
 The preprocessing errs with probability $\Oh(1/n)$ for each $u_{i,r}$, i.e., $\Oh(sw/n)$ in total.
\end{theorem}

Let us defer the proof of the theorem until we complete the description of the algorithm and derive Theorem~\ref{th:main}. 
We preprocess functions $u_{i,r}$ and build a trie on $h(S_1), \dots, h(S_{2n})$ for each $h\in \Hashes_i$. 
We then augment the trie with an LCA data structure, which can be done in linear time and space~\cite{DBLP:journals/siamcomp/FischerH11,DBLP:journals/siamcomp/HarelT84}. 
The latter can be used to find in constant time the longest common prefix of any two strings $h(S_j)$ and $h(S_{j'})$.

Consider a function $h\in\Hashes$ and a positive integer $\ell\le n$.
We define $\restr{h}{[\ell]}$ so that
$$\restr{h}{[\ell]}(S) =\begin{cases}
h(S[1,\ell]) & \text{if }|S|\ge \ell,\\
h(S) & \text {otherwise.}
\end{cases}$$
In other words, if $h$ is a projection onto positions from a multiset $\Pos$, then $\restr{h}{[\ell]}$
is a projection onto positions from the multiset $\{p\in \Pos : p \le \ell\}$, extended with $ \$$'s to length $mt$.
Consequently, $\restr{h}{[\ell]}(S)=\restr{h}{[\ell]}(S')$ if and only if 
the longest common prefix of $h(S)$ and $h(S')$
is at least $|\{p\in \Pos : p \le \ell\}|$ characters long.

We define the family of \emph{collisions} $\Collisions^{\Hashes}_\ell$ as a set of triples $(S,S',h)$ such that $S$ and $S'$
are suffixes of $T_1$ and $T_2$, respectively, both of length at least $\ell$, and
$h\in \Hashes$ is such that the suffixes collide on $\restr{h}{[\ell]}$, that is, $\restr{h}{[\ell]}(S)=\restr{h}{[\ell]}(S')$.
Note that the families of collisions are nested: $\Collisions^{\Hashes}_{0}\supseteq \cdots \supseteq \Collisions^{\Hashes}_{\ell}\supseteq \Collisions^{\Hashes}_{\ell+1}\supseteq\cdots \supseteq \Collisions^{\Hashes}_{n}$.

\begin{algorithm}[ht]
\caption{Longest common substring with approximately $k$ mismatches.}
\begin{algorithmic}[1]
\State Preprocess $T_1, T_2$ for $\lcp$ queries

\For {$i = 1, 2, \ldots , s$}
	\For {$r = 1, 2, \ldots , w$}
		\State Choose a function $u_{i,r} \in \Projections^m$ uniformly at random and preprocess it using Theorem~\ref{th:sort}
	\EndFor
	\ForAll {$h = (u_{i,r_1}, u_{i,r_2}, \ldots, u_{i,r_t})$} 
		\State Build a trie on $h(S_1), h(S_2), \ldots, h(S_{2n})$ and augment it with an LCA data structure
		\State Add $h$ to $\Hashes$
	\EndFor	
\EndFor
		\State Find the largest $\ell$ such $|\Collisions_\ell^{\Hashes}|\ge 2n|\Hashes|$
		\ForAll {$(S,S',h)\in \Collisions_{\ell+1}^{\Hashes}$}
			 \State Compute $\lcp (S, S')$ and update the answer
		\EndFor
		\State Pick $(\bar{S},\bar{S}',\bar{h})\in \Collisions_{\ell}^{\Hashes}$ uniformly at random
		\State Compute $\lcpe (\bar{S}, \bar{S}')$ and update the answer 
\end{algorithmic}
\label{alg:LSH}
\end{algorithm}

For a fixed function $h$, we define the \emph{$\ell$-neighbourhood} of $S$ as the set of suffixes $S'$ of $T_2$
such that $(S,S',h)\in \Collisions^{\Hashes}_\ell$.
We observe that the $\ell$-neighbourhood of $S$ forms a contiguous range in the sequence of suffixes of $T_2$
ordered according to $\prec_h$, and this range can be identified in $\Oh(\log n)$ time using binary search and LCA queries on the trie constructed for $h$.
Consequently, an $\Oh(n|\Hashes|)$-space representation of $\Collisions^{\Hashes}_\ell$,
with one range for every $\ell$-neighbourhood of each suffix $S$, can be constructed in $\Oh(n|\Hashes|\log n)$ time.

In the algorithm, we find the largest $\ell$ such that $|\Collisions^{\Hashes}_{\ell}|\ge 2n|\Hashes|$; using a binary search,
this takes $\Oh(n|\Hashes|\log^2 n)$ time.
For each $(S,S',h)\in \Collisions^{\Hashes}_{\ell+1}$, we compute the longest common prefix with approximately $k$ mismatches $\lcp(S,S')$ (Lemma~\ref{lm:lcp_k}).
Additionally, we pick a single element $(\bar{S},\bar{S}',\bar{h})\in \Collisions^{\Hashes}_{\ell}$ uniformly at random
and compute the longest common prefix with at most $(1+\eps)k$ mismatches $\lcpe(\bar{S},\bar{S}')$ naively in $\Oh(n)$ time.
The longest of the retrieved prefixes is returned as an answer. 

See Algorithm~\ref{alg:LSH} for pseudocode. We will now proceed to the analysis of complexity and correctness of the algorithm.

\subsection{Complexity and correctness} 
To ensure the complexity bounds and correctness of the algorithm, we must carefully choose the parameters $s$, $t$, $w$, and $m$. Let $p_1 = 1 - k / n$, $p_2 = 1 - (1+\eps) \cdot k / n$,
and $\rho =\log p_1 / \log p_2$. The intuition behind these values is that if $S$ and $S'$ are two strings of length $n$ and $\HD(S,S')\le k$, then 
$p_1$ is a lower bound for the probability of $S[i]=S'[i]$ for a uniformly random position $i$. On the other hand, $p_2$ is an upper bound for the same probability
if $\HD(S,S')\ge (1+\eps)\cdot k$. Based on these values, we define
$$t = \big\lceil \sqrt{\log n} \, \big\rceil, \quad m = \left \lceil \tfrac{1}{t} \log_{p_2}{\tfrac{1}{n}} \right \rceil,
\quad w = t^2 + \lceil p_1^{-m} \rceil,\;\text{ and }\; s=\Theta(t!).$$
We assume that $(1+\eps)k<n$ in order to guarantee $p_1>p_2>0$. Note that if $(1+\eps)k \ge n$, the problem is trivial.

\subsubsection{Complexity} To show the complexity of the algorithm, we will start with a simple observation
and a more involved fact.

\begin{observation}\label{obs:ws}
We have $s = n^{o(1)}$ and $w=n^{o(1)}$.
\end{observation}
\begin{proof}
First, observe
$$s = \Oh(t!) = 2^{\Oh(t \log t)}=2^{\Oh(\sqrt{\log n} \log \log n)}=2^{o(\log n)}=n^{o(1)}.$$
Similarly,
$$w = t^2 + \lceil{p_1^{-m}}\rceil \le t^2 + 1 + p_1^{-m}  = \Oh(\log n)+p_1^{-\Oh(\frac{1}{t}\log_{p_2}\frac1n)}
= \Oh(\log n) + 2^{\Oh(\rho \sqrt{\log n})}.$$
Moreover, $p_1 > p_2$ yields $\log p_1 > \log p_2$
and therefore $\rho = \frac{\log p_1}{\log p_2} < 1$.
Consequently, $w = \Oh(\log n) + 2^{\Oh(\sqrt{\log n})} = n^{o(1)}$,
which concludes the proof.
\end{proof}

\begin{fact}\label{ft:h}
We have $|\Hashes| = \Oh(n^{1/(1+\eps)})$.
\end{fact}
\begin{proof}
Observe that $|\Hashes| = s \tbinom{w}{t} = \Oh(t! \tbinom{w}{t}) = \Oh(w^t)$.
To estimate the latter, we consider two cases.
If $w \le 3t^3$,
then $$w^t = \big(3t^3\big)^{t}=2^{\Oh(t\log t)} = 2^{\Oh(\sqrt{\log n}\log \log n)} = 2^{o(\log n)} = n^{o(1)} = \Oh(n^{1/(1+\eps)}).$$
Otherwise,
 $p_1^{-m} \ge w-1-t^2 \ge 3t^3-t^2\ge t^3+t$.
 Consequently,
 $$w^t = (t^2 + \lceil{p_1^{-m}}\rceil)^{t}\le (t^2+1+p_1^{-m})^t = p_1^{-mt}\left(1+\tfrac{t^2+1}{p_1^{-m}}\right)^t\le p_1^{-mt}\left(1+\tfrac{1}{t}\right)^t \le p_1^{-mt}\cdot e.$$
 Thus, it suffices to prove that $p_1^{-mt} = \Oh(n^{1/(1+\eps)})$.
 We have
 $$\log (p_1^{-mt})= -t\left\lceil\tfrac{1}{t}\log_{p_2}\tfrac{1}{n}\right\rceil\log p_1\le (-\log_{p_2}\tfrac{1}{n}-t)\log p_1 =\rho \log n-t\log p_1.$$
 Moreover, due to $(1+\eps)k < n$ and $\eps = \Theta(1)$, we have
 $$-t\log p_1 = -t \log(1-\tfrac{k}{n})=t\log\tfrac{n}{n-k}=\Oh\big(t\tfrac{k}{n-k}\big)=\Oh\big(t\tfrac{1+\eps}{\eps}\tfrac{k}{n}\big)=O\big(\tfrac{k}{n}\sqrt{\log n}\big).$$
 On the other hand, taking the Taylor's expansion of $f(x)=\tfrac{\log (1-x)}{\log (1-(1+\eps)x)}$, which is concave for $0\le x < \frac{1}{1+\eps}$,
 we obtain
 $$\rho = \tfrac{\log \left(1-\tfrac{k}{n}\right)}{\log \left(1-(1+\eps)\tfrac{k}{n}\right)} \le \tfrac{1}{1+\eps}-\tfrac{\eps}{2(\eps+1)}\tfrac{k}{n} = \tfrac{1}{1+\eps} - \Theta(\tfrac{k}{n}).$$
Consequently, 
$$\log (p_1^{-mt}) \le \rho \log n-t\log p_1\le  \tfrac{\log n}{1+\eps} - \Theta\big(\tfrac{k}{n}\log n\big) + O\big(\tfrac{k}{n}\sqrt{\log n}\big)\!=\! \tfrac{\log n}{1+\eps}-\Theta\big(\tfrac{k}{n}\log n\big). $$
Thus, $p_1^{-mt} \le n^{1/(1+\eps)}$ holds for sufficiently large $n$
and therefore
$|\Hashes| = \Oh(w^t) = \Oh(e \cdot p_1^{-mt}) = \Oh(n^{1/(1+\eps)}),$
which concludes the proof.
\end{proof}

\begin{lemma}\label{lm:time}
The running time of the algorithm is $\Oh(n^{1+1/(1+\eps)} \log^2 n)$.
\end{lemma} 
\begin{proof}
Preprocessing for $\lcp$ queries takes  $\Oh (n \log^3 n)$ time (Lemma~\ref{lm:lcp_k}),
whereas functions $u_{i,r}$ are processed in $\Oh(ws\cdot n^{4/3} \log^{4/3} n)$ overall time
using Theorem~\ref{th:sort}.
Afterwards, for each hash function $h\in\Hashes$ we can build a trie and an LCA data structure on strings $h(S_1),\ldots ,h(S_{2n})$ in $\Oh (t n \log n)$ time,
which is $\Oh(|\Hashes|tn\log n)$ in total.
Next, the value $\ell$ and the family $|\Collisions_{\ell+1}^{\Hashes}|$ are computed in $\Oh(|\Hashes|n\log^2 n)$ time.
The time for $|\Collisions_{\ell+1}^{\Hashes}| < 2 n |\Hashes|$ $\lcp$ queries is bounded by the same function.
Finally, we answer one $\lcpe$ query, which takes $\Oh(n)$ time.
The overall running time is $$\Oh(n\log^3 n +ws\cdot n^{4/3}\log^{4/3 }n +\! |\Hashes|n\log n(t+\log n))\!=\!\Oh(n^{4/3+o(1)}+n^{1+1/(1+\eps)}\log^2 n)$$
due to Observation~\ref{obs:ws} and Fact~\ref{ft:h}.
We can hide the first term because of $\eps < 2$.\end{proof}

\begin{lemma}
The space complexity of the algorithm is $\Oh (n^{1+1/(1+\eps)})$. 
\end{lemma}
\begin{proof}
The data structure for $\lcp$ queries requires $\Oh(n \log^2 n)$ space. Preprocessing functions $u_{i,r}$ requires $\Oh (sw n) = \Oh (n^{1+o(1)})$ space and the tries occupy $\Oh(|\Hashes| \cdot n) = \Oh (n^{1+1/(1+\eps)})$ space.
\end{proof}

\subsubsection{Correctness} 
First, let us focus on two suffixes which yield the longest common substring with exactly $k$ mismatches.

\begin{lemma}\label{lm:hash_function_exists}
Let $S$ and $S'$ be suffixes of $T_1$ and $T_2$, respectively, that maximise $\lcpk(S,S')$,
i.e., such that $\lcpk(S,S')=\ell_k$.
For each $i\in\{1,\ldots,s\}$, with probability $\Omega(1/t!)$ there exists $h\in \Hashes_i$ such that $\restr{h}{[\ell_k]}(S)=\restr{h}{[\ell_k]}(S')$.
\end{lemma}
\begin{proof}
By definition of $\ell_k$, we have $\HD(S[1,\ell_k],S'[1,\ell_k])\le k$.
Moreover, for any hash function $h$ we have that  $\restr{h}{[\ell_k]} (S) = \restr{h}{[\ell_k]}(S')$ if and only if $h(S[1,\ell_k]) = h(S'[1,\ell_k])$. Let us recall that each hash function $h\in\Hashes_i$ is a $t$-tuple of functions $u_{i,r}\in \Projections^m$. Consequently, $h(S[1,\ell_k]) = h(S'[1,\ell_k])$ for some $h\in \Hashes_i$ if and only if the strings 
$S[1,\ell_k] \$^{n-\ell_k}$ and $S'[1,\ell_k] \$^{n-\ell_k}$ collide on at least $t$ out of $w$ functions $u_{i,r}$. 
We shall give a lower bound on the probability $\mu$ of this event.
Individual collisions are independent and each of them holds with the same probability $q=p_1^m$.
Moreover, $\mu$ may only increase as we increase $q$, so we can replace $q$ by a lower bound $\frac{1}{w}$.
(Note that $w = t^2+ \lceil{p_1^{-m}}\rceil \ge \lceil{q^{-1}\rceil}\ge q^{-1}$.)
We have
$$\mu = \sum_{i=t}^w \tbinom{w}{i} q^i (1-q)^{w-i} \ge \sum_{i=t}^w \tbinom{w}{i} \tfrac{1}{w^i} (1-\tfrac{1}{w})^{w-i}
\ge \tbinom{w}{t}\tfrac{1}{w^t} (1-\tfrac{1}{w})^{w} \ge\tfrac{1}{t!}(\tfrac{w-t+1}{w})^t(\tfrac{w-1}{w})^w.$$
Hence,
$$
\tfrac{1}{\mu t!}\le (\tfrac{w}{w-t+1})^t(\tfrac{w}{w-1})^w=(1+\tfrac{t-1}{w-t+1})^t(1+\tfrac{1}{w-1})^w \le \exp\big(\tfrac{t(t-1)}{w-t+1}+\tfrac{w}{w-1}\big)=\Oh(1),$$
where the latter is true because $w\ge t^2$ and $w\ge 2$.  Consequently, $\mu = \Omega(1/t!)$.
\end{proof}
As a corollary, we can choose a constant in the number of steps $s=\Theta(t!)$
so that $(S,S',h)\in \Collisions^{\Hashes}_{\ell_k}$ for some $h\in \Hashes$ holds with probability at least $\frac34$.
If additionally $\ell_k > \ell$, then $(S,S',h) \in \Collisions^{\Hashes}_{\ell+1}$,
so $\lcp(S,S')$ will be called and with high probability will return a substring of length $\ge \ell_k$.
Otherwise, $|\Collisions_{\ell_k}^{\Hashes}|\ge 2n|\Hashes|$
and we claim that  a uniformly random $(\bar{S},\bar{S}',\bar{h})\in \Collisions_{\ell}^{\Hashes}$
satisfies $\lcpe(\bar{S},\bar{S}')\ge \ell \ge \ell_k$ with probability at least~$\frac12$.
To prove this, we first introduce a family $\Bad^{\Hashes}$ of \emph{bad collisions}:
triples $(S,S',h)$ which belong to $\Collisions_{\ell}^{\Hashes}$ for some $\ell>\lcpe(S,S')$,
and bound its expected size.

\begin{lemma}\label{lm:collisions}
The expected number of bad collisions satisfies $\mathbb{E}[|\Bad^{\Hashes}|] \le n|\Hashes|$.
\end{lemma}
\begin{proof}
Let us bound the probability that $(S,S',h)\in \Bad^{\Hashes}$ for fixed suffixes $S$ and $S'$ (of $T_1$ and $T_2$, respectively)
and fixed $h=(u_{i,r_1},\ldots,u_{i,r_t})$. Equivalently, we shall bound $\Prob[(S,S',h)\in \Collisions^{\Hashes}_{\ell}]$ for $\ell=\lcpe(S,S')+1$.

If $|S|<\ell$ or $|S'|<\ell$, the probability is 0 by the definition of $\Collisions^{\Hashes}_{\ell}$.
Otherwise, we observe that $d_H(S[1,\ell],S'[1,\ell])>(1+\eps)k$ and that $h$ can be considered (due to its marginal distribution) as a projection onto $mt$ uniformly random positions.
Therefore,
$$\Prob [\restr{h}{[\ell]} (S) = \restr{h}{[\ell]} (S')] = \Prob [h(S[1,\ell]\$^{n-\ell}) = h(S'[1,\ell]\$^{n-\ell})] \le p_2^{mt}\le \tfrac{1}{n},$$ 
where the last inequality follows from the definition of $m$, which yields $mt \ge \log_{p_2} \frac{1}{n}$.

In total, we have $n^2 |\Hashes|$ possible triples $(S,S',h)$ so by linearity of expectation, we conclude that the expected
number of bad collisions is at most $\frac{1}{n}n^2 |\Hashes|=n|\Hashes|$.
\end{proof}

\begin{corollary}\label{cr:collisions}
Let $(S,S',h)$ be a uniformly random element of $\Collisions_{\ell}^{\Hashes}$,
where $\ell$ is a \emph{random variable} which always satisfies $|\Collisions_\ell^{\Hashes}|\ge 2n|\Hashes|$.
We have $\Prob[(S,S',h)\in \Bad^{\Hashes}] \le \frac12$.
\end{corollary}
\begin{proof}
More formally, we shall prove that $\Pr[(S,S',h)\in \Bad^{\Hashes} \mid  (S,S',h)\in \Collisions_{\ell}^{\Hashes}]\le \frac12$
holds for a uniformly random triple $(S,S',h)$. 
Indeed:
\[\Pr[(S,S',h)\in \Bad^{\Hashes} \mid  (S,S',h)\in \Collisions_{\ell}^{\Hashes}]=\mathbb{E}\left[\tfrac{|\Bad^{\Hashes}\cap \Collisions_{\ell}^{\Hashes}|}{|\Collisions_{\ell}^{\Hashes}|}\right]\le  \mathbb{E}\left[\tfrac{|\Bad^{\Hashes}|}{2n|\Hashes|}\right]\le \tfrac12.\quad \qedhere\]
\end{proof}

Below, we combine the previous results to prove that with constant probability Algorithm~1 correctly solves
the \ApproxLCS problem.
Note that we can reduce the error probability to an arbitrarily small constant $\delta>0$:
it suffices to repeat the algorithm a constant number of times and among the resulting pairs,
choose the longest substrings successfully verified to be at Hamming distance at most $(1+\eps)k$;
verification can be implemented naively in $\Oh(n)$ time.

\begin{corollary}
With non-zero constant probability, Algorithm 1 succeeds~--- it reports a substring of $T_1$ 
and a substring of $T_2$ at Hamming distance at most $(1+\eps)k$, both of length at least $\ell_k$, where $\ell_k$ is the length of the longest common substring with $k$ mismatches.
\end{corollary}
\begin{proof}
We will prove that the algorithm succeeds conditioned on the following events:
\begin{itemize}
  \item the preprocessing of Lemma~\ref{lm:lcp_k} succeeds,
  \item the preprocessing of Theorem~\ref{th:sort} succeeds for each function $u_{i,r}$,
  \item $\Collisions_{\ell_k}^{\Hashes}$ contains $(S,S',h)$ such that $\lcpk(S,S')= \ell_k$ (see Lemma~\ref{lm:hash_function_exists}),
  \item the randomly chosen $(\bar{S},\bar{S}',\bar{h})\in \Collisions_{\ell}^{\Hashes}$ does not belong to $\Bad^{\Hashes}$ (see Corollary~\ref{cr:collisions}).
\end{itemize}
This assumption holds with probability $\Omega(1)$, because probability of the complementary event 
can be bounded as follows using the union bound  applied on the top of Lemma~\ref{lm:lcp_k}, Theorem~\ref{th:sort}, Lemma~\ref{lm:hash_function_exists}, and Corollary~\ref{cr:collisions}:
$$\tfrac{1}{n^3}+\Oh(\tfrac{ws}{n})+\tfrac14+\tfrac12=\tfrac34+o(1)=1-\Omega(1).$$ 

Successful preprocessing of functions $u_{i,r}$ guarantees that the value $\ell$ and the families $\Collisions_{\ell}^{\Hashes}$ and $\Collisions_{\ell+1}^{\Hashes}$
have been computed correctly. If $\ell_k > \ell$, then $\Collisions_{\ell+1}^{\Hashes}$ contains $(S,S',h)$ such that $\lcpk(S,S')=\ell_k$.
The correctness of $\lcp$ queries asserts that $\lcp(S,S')\ge \ell_k$, so the algorithm considers prefixes of $S$ and $S'$ of length at least $\ell_k$
as candidates for the resulting substrings. 
If $\ell_k \le \ell$, on the other hand, then the randomly chosen $(\bar{S},\bar{S}',\bar{h})\in \Collisions_{\ell}^{\Hashes}$  satisfies $\lcpe(\bar{S},\bar{S}')\ge \ell\ge \ell_k$,
so the algorithm considers prefixes of $\bar{S}$ and $\bar{S}'$ of length at least $\ell\ge \ell_k$.
In either case, a pair of substrings of length at least $\ell_k$ and at Hamming distance at most $(1+\eps)k$ is among the considered candidates.
The resulting substrings also satisfy these conditions, because we return the longest candidates and the correctness of $\lcp$ queries asserts that no substrings at distance more than $(1+\eps)k$
are considered. 
\end{proof}

\subsection{Proof of Theorem~\ref{th:sort}}
Recall that each $h\in \Hashes$ is a $t$-tuple of functions $u_{i,r}$, i.e.\ $h = (u_{i,r_1}, u_{i,r_2}, \ldots, u_{i,r_t})$, where $1\le i \le s$ and $1 \le r_1 < r_2 < \cdots < r_t \le w$. We will show a preprocessing of functions $u_{i,r}$ after which we will be able to compute the longest common prefix of any two strings $u_{i,r} (S_j), u_{i,r}(S_{j'})$ in $\Oh(1)$ time. As a result, we will be able to compute the longest common prefix of $h(S_j), h(S_{j'})$ in $\Oh(t)$ time. It also follows that we will be able to compare any two strings $h(S_j), h(S_{j'})$ in $\Oh(t)$ time as the order $\prec_h$ is defined by the character following the longest common prefix (or by the lengths $|S_j|$ and $|S_{j'}|$ if $h(S_j)=h(S_{j'})$). Therefore, we can sort strings $h(S_1), h(S_2), \ldots, h(S_{2n})$ in $\Oh (t n \log n)$ time and $\Oh(n)$ space and then compute the longest common prefix of each two adjacent strings in $\Oh(tn)$ time.  The trie on $h(S_1), h(S_2), \ldots, h(S_{2n})$ can then be built in $\Oh (n)$ time by imitating its depth-first traverse. 

It remains to explain how we preprocess individual functions $u_{i,r}$. For each function, it suffices to build a trie on strings $u_{i,r}(S_1), u_{i,r}(S_2), \ldots, u_{i,r}(S_{2n})$ and to augment it with an LCA data structure~\cite{DBLP:journals/siamcomp/FischerH11,DBLP:journals/siamcomp/HarelT84}. We will consider two different methods for constructing the trie with time dependent on~$m$. No matter what the value of $m$ is, one of these methods will have $\Oh(n^{4/3}\log^{4/3} n)$ running time. Let $u_{i,r}$ be a projection onto a multiset $\Pos$ of positions $1 \le a_1 \le a_2 \le \cdots \le a_m \le n$ and denote $T = T_1 \$^{n} T_2 \$^{n}$. 

\begin{lemma}\label{lm:sort1}
The trie on $u_{i,r}(S_1), \dots, u_{i,r}(S_{2n})$ can be constructed in $\Oh (\sqrt{m} n \log n)$ time and $\Oh(n)$ space correctly with error probability at most $1/n$.
\end{lemma}
\begin{proof}
Without loss of generality assume that $\sqrt{m}$ is integer. Let us partition $\Pos$ into subsets $B_1, \dots, B_{\sqrt{m}}$, where 

$$B_\ell = \{a_{\ell,1}, a_{\ell,2}, \ldots, a_{\ell, \sqrt{m}}\} = \{a_{(\ell-1)\sqrt{m}+q} \; | \; q \in [1, \sqrt{m}]\}.$$ 

Now $u_{i,r}$ can be represented as a $\sqrt{m}$-tuple of projections $b_1, b_2, \ldots, b_{\sqrt{m}}$ onto the subsets $B_1, B_2, \ldots, B_{\sqrt{m}}$, respectively. We will build the trie by layers to avoid space overhead. Suppose that we have built the trie for a function $(b_1, b_2, \ldots, b_{\ell-1})$ and we want to extend it to the trie for $(b_1, b_2, \ldots, b_{\ell-1}, b_\ell)$. 

Let $p$ be a prime of value $\Omega(n^5)$. With error probability inverse polynomial in $n$, we can find such $p$ in $\Oh(\log^{\Oh(1)}n)$ time;
see~\cite{DBLP:journals/moc/TaoCH12,agrawal2004primes}. 
We choose a uniformly random $r \in \mathbb{F}_p$ and create a vector $\chi$ of length $n$.
We initialise $\chi$ as a zero vector and for each position $a_{\ell,q}\in B_\ell$, we increase $\chi [a_{\ell,q}]$ by $r^{q}$.
We then run the FFT algorithm for $\chi$ and $T$ in the field $\mathbb{Z}_p$~\cite{FischerPaterson}. The output of the FFT algorithm contains the inner products of $\chi$ and all suffixes $S_1, S_2, \ldots, S_{2n}$. The inner product of $\chi$ and a suffix $S_j$ is the Karp--Rabin fingerprint~\cite{DBLP:journals/ibmrd/KarpR87} $\varphi_{\ell,j}$ of $b_\ell (S_j)$, where

$$\varphi_{\ell,j} = \sum_{q=1}^{\sqrt{m}} S_j [a_{\ell,q}] \cdot r^{q} \pmod p.$$

If the fingerprints of $b_\ell (S_j)$ and $b_\ell (S_{j'})$ are equal, then $b_\ell (S_j)$ and  $b_\ell (S_{j'})$ are equal with probability at least $1-1/n^4$, and otherwise they differ (for a proof, see e.g.~\cite{Porat:09}).

For a fixed leaf of the trie for $(b_1, b_2, \ldots, b_{\ell-1})$, we first sort all the suffixes that end in it by fingerprints~$\varphi_{\ell,j}$. Second, we lexicographically sort the strings $b_\ell(S_{j})$ with distinct fingerprints. For this, we need to be able to compare $b_\ell (S_j)$ and  $b_\ell (S_{j'})$ and to find the first character where they differ. We compare $b_\ell (S_j)$ and  $b_\ell (S_{j'})$ character-by-character in $\Oh(\sqrt{m})$ time. We then append the leaf of the trie for $(b_1, b_2, \ldots, b_{\ell-1})$ with a trie on strings $b_\ell (S_j)$ that can be built by imitating its depth-first traverse.

By the union bound, the error probability is at most $\frac{1}{n^4} \cdot n^2 \sqrt{m} \le \frac{1}{n}$.
We now analyse the complexity of the algorithm. For each of the $\sqrt{m}$ layers, the FFT algorithm takes $\Oh(n \log n)$ time. The sort by fingerprints takes $\Oh(n \log n)$ time per layer, or $\Oh(\sqrt{m} n \log n)$ time in total. We finally need to estimate the total number of character-by-character comparisons in all the layers. We claim that it can be upper bounded by $\Oh(n \log n)$. The reason for that is as follows: if we consider the resulting trie for $u_{i,r}(S_1), \dots, u_{i,r}(S_{2n})$, it has size $\Oh(n)$. Imagine that the layers cut this trie into a number of smaller tries. The total size of these tries is still $\Oh(n)$, and we build each of these tries using character-by-character comparisons. For a trie of size $x$, we need $\Oh(x \log x)$ comparisons, which in total is $\Oh(n \log n)$. Therefore, the character-by-character comparisons take $\Oh(\sqrt{m} n \log n)$ time in total.
\end{proof}

The second method builds the trie using the algorithm described in the first paragraph of this section: we only need to give a method for computing the longest common prefix of $u_{i,r} (S_j)$ and $u_{i,r}(S_{j'})$ (or, equivalently, the first position where $u_{i,r} (S_j)$ and $u_{i,r}(S_{j'})$ differ). The following lemma shows that this query can be answered in $\Oh (n \log n / m)$ time, which gives $\Oh (n^2 \log^2 n / m)$ time complexity of the trie construction.

\begin{lemma}[see \cite{substringNN}]\label{lm:sort2}
After $\Oh(n)$-time and space preprocessing the first position where two strings $u_{i,r}(S_j)$ and $u_{i,r}(S_{j'})$ differ can be found in $\Oh (n \log n / m)$ time correctly with error probability at most $1/n^3$.
\end{lemma}
\begin{proof}
For $m=\Oh(\log n)$ the conclusion is trivial.
Assume otherwise.
We start by building the suffix tree for the string $T$ which takes $\Oh(n)$ time and space~\cite{DBLP:conf/focs/Weiner73,DBLP:books/cu/Gusfield1997}. Furthermore, we augment the suffix tree with an LCA data structure in $\Oh(n)$ time~\cite{DBLP:journals/siamcomp/FischerH11,DBLP:journals/siamcomp/HarelT84}. 

Let $\ell =  \lceil 3 n \ln n / m \rceil$.  We can find the first $\ell$ positions $q_1 < q_2 < \cdots < q_\ell$ where $S_{j}$ and $S_{j'}$ differ in $\Oh(\ell)=\Oh(n \log n / m)$ time using the kangaroo method~\cite{DBLP:journals/tcs/LandauV86,kangaroo-2}. We set $q_r=\infty$ if a given position does not exist. The idea of the kangaroo method is as follows. We can find $q_1$ by one query to the LCA data structure in $\Oh(1)$ time. After removing the first $q_1$ positions of $S_{j}$ and $S_{j'}$, we obtain suffixes $S_{j+q_1}, S_{j'+q_1}$ and find $q_2$ by another query to the LCA data structure, and so on. If at least one of the positions $q_1, q_2, \ldots, q_\ell$ belongs to  $\Pos$, then we return the first such position as an answer, and otherwise we say that $u_{i,r} (S_j) = u_{i,r} (S_{j'})$. The multiset $\Pos$ can be stored as an array of multiplicities so that testing if an element belongs to it can be done in constant time.

Let us show that if $p$ is the first position where $u_{i,r}(S_j)$ and $u_{i,r}(S_{j'})$ differ, then $p$ belongs to $\{q_1, q_2, \ldots, q_\ell\}$ with high probability. Because $q_1 < q_2 < \cdots < q_\ell$ are the first $\ell$ positions where $S_j$ and $S_{j'}$ differ, it suffices to show that at least one of these positions belongs to $\Pos$. We rely on the fact that positions of $\Pos$ are independent and uniformly random elements of $[1,n]$. Consequently, we have $\Pr[q_1, \ldots, q_\ell \notin \Pos] = (1 - \ell/n)^{m} \le (1 - 3 \ln n/m)^{m} \le \frac{1}{e^{3\ln n}} = 1/n^3$. 
\end{proof}

By Lemmas~\ref{lm:sort1} and~\ref{lm:sort2}, the trie on strings $u_{i,r} (S_1), \dots, u_{i,r} (S_{2n})$ can be built in $\Oh(\min \{\sqrt{m}, n \log n / m \} \cdot n \log n) = \Oh(n^{4/3} \log^{4/3} n)$ time and $\Oh(n)$ space correctly with high probability which implies Theorem~\ref{th:sort} as explained in the beginning of this section.

\section{\ApproxLCS}
\label{sec:ApproxLCS}	
	In this section, we consider an approximate variant of the \kLCS problem, defined as follows.

\begin{problem}[\ApproxLCS]\label{pr:apprLCS}
Two strings $T_1, T_2$ of length $n$, an integer $k$, and a constant $z > 1$ are given. If $\ell_k$ is the length of the longest common substring with $k$ mismatches of $T_1$ and $T_2$, return a substring of $T_1$ of length at least $\ell_k/z$ that occurs in $T_2$ with at most $k$ mismatches.
\end{problem}

\begin{theorem}\quad
  \begin{enumerate}[(a)]
    \item  The \ApproxLCS problem for $z=2$ can be solved in $\Oh(n^{1.5}\log^2 n)$ time and $\Oh(n^{1.5})$ space.
    \item Suppose there exist $0 < \eps < 1$ and $\delta>0$ such that the \ApproxLCS problem for $z=2-\eps$ and a binary alphabet can be solved in $\Oh(n^{2-\delta})$ time. Then SETH is false.
  \end{enumerate}
\end{theorem}
\begin{proof}
  (a) The algorithm of Theorem~\ref{th:main} for $\eps=1$ computes a pair of substrings of length at least $\ell_k$ of $T_1$ and $T_2$ that have Hamming distance at most $2k$. Either the first halves or the second halves of the strings have Hamming distance at most $k$.

  (b) We use the gap that exists in Lemma~\ref{lem:red} for $q>1$. Assume that there is such an algorithm for some $\eps$ and $\delta$. We will run it for strings $T_1$ and $T_2$ from that lemma. Let $q = \lceil\frac{3}{\eps}\rceil-2$; then $\ell/\ell' \ge 2-\eps$. If the \OV problem has a solution, by Lemma~\ref{lem:red}\eqref{it:a}, the algorithm produces a longest common substring of length at least $\ell/(2-\eps) \ge \ell'$. Otherwise, by Lemma~\ref{lem:red}\eqref{it:b}, its result has length smaller than $\ell'$.
This concludes that the conjectured approximation algorithm can be used to solve the \OV problem. 
  
The lengths of the selected strings are $n=N(7dq+7d)+7dq=\Oh(Nd)$ for $d=c \log N$. Hence, the running time is $\Oh(n^{2-\delta})=\Oh(N^{2-\delta}d^{\Oh(1)})$, which, by Fact~\ref{fct:OVP}, contradicts SETH.
\end{proof}

\section{\AllkLCS}
\label{sec:AllkLCS}
  The following problem has received a considerable attention in the recent years; see \cite{DBLP:conf/stoc/ChanL15}
  and the references therein.

  \begin{problem}[\BJI]~\label{pr:BJI}
    Construct a data structure over a binary string $S$ of length $n$ that, given positive integers $\ell$ and $q$,
    can compute if there is a substring of $S$ of length $\ell$ containing exactly $q$ ones.
  \end{problem}

  A simple combinatorial argument shows that it suffices to compute the minimal and maximal number of ones
  in a substring of $S$ of length $\ell$, as for every intermediate number of ones a substring of $S$ of this length
  exists as well.
  As a result, the \BJI problem can be solved in linear space and with constant-time queries.
  It turns out that the index can also be constructed in strongly subquadratic time.

  \begin{lemma}[Chan and Lewenstein \cite{DBLP:conf/stoc/ChanL15}]\label{lem:BIJ}
    The index for \BJI of $\Oh(n)$ size and with $\Oh(1)$-time queries can be constructed in $\Oh(n^{1.859})$ expected time or in $\Oh(n^{1.864})$ deterministic time.
  \end{lemma}

  We use this result to solve the \kLCS problem for all values of $k$ simultaneously.

  \begin{theorem}
    \AllkLCS can be solved in $\Oh(n^{2.859})$ expected time or in $\Oh(n^{2.864})$ deterministic time.
  \end{theorem}
  \begin{proof}
    Note that, equivalently, we can compute, for all $\ell=1,\ldots,n$, the minimal Hamming distance between substrings of length $\ell$ in $T_1$ and $T_2$.

    Let $M$ be an $n \times n$ Boolean matrix such that $M[i,j]=0$ if and only if $T_1[i]=T_2[j]$.
    We construct $2n-1$ binary strings corresponding to the diagonals of $M$: the string number $p$,
    for $p \in \{-n,\ldots,n\}$, corresponds to the diagonal $M[i,j]\,:\,j-i=p$.
    For each of the strings, we construct the jumbled index using Lemma~\ref{lem:BIJ}.

    Each diagonal corresponds to one of the possible alignments of $T_1$ and $T_2$.
    In the jumbled index we compute, in particular, for each value of $\ell$
    what is the minimal number of 1s (which correspond to mismatches between
    the corresponding positions in $T_1$ and $T_2$) in a string of length $\ell$.
    To compute the global minimum for a given $\ell$, we only need to take the minimum
    across all the jumbled indexes.

    By Lemma~\ref{lem:BIJ}, all the jumbled indexes can be constructed in $\Oh(n^{2.859})$ expected time
    or in $\Oh(n^{2.864})$ time deterministically.
  \end{proof}

\bibliographystyle{abbrv}
\bibliography{main}
\end{document}